\documentclass[12pt]{article}

\usepackage{graphicx}
\usepackage{amsmath}
\usepackage{amsfonts}
\usepackage{amssymb}
\usepackage{amsthm}
\usepackage[all,cmtip]{xy}

\usepackage{amsmath, amsthm, amssymb}

\theoremstyle{plain}	

\newtheorem{thm}{Theorem}[section]
\newtheorem{lemma}[thm]{Lemma}
\newtheorem{prop}[thm]{Proposition}

\theoremstyle{definition}

\newtheorem{defn}[thm]{Definition}
\newtheorem{examp}[thm]{Example}

\begin{document}

\title{ The Ponzano-Regge model and parametric representation}
\author{\textsc{Dan Li}
				\thanks{
				Department of Mathematics, Florida State University, 
				Tallahassee, FL 32306. 
				\newline \indent ~ E-mail: dli@math.fsu.edu
			        	}
	    }

\date{}
\maketitle

\begin{abstract}
We give a parametric representation of the effective noncommutative field theory derived from a $\kappa$-deformation of the Ponzano-Regge model and define a generalized Kirchhoff polynomial with $\kappa$-correction terms, obtained in a $\kappa$-linear approximation. We then consider the corresponding graph hypersurfaces and the question of how the presence of the correction term affects their motivic nature. We look in particular at the tetrahedron graph, which is the 
basic case of relevance to quantum gravity. With the help of computer calculations, we verify that the number of points over finite fields of the corresponding hypersurface does not fit polynomials with integer coefficients, hence the hypersurface of the tetrahedron is not polynomially countable. 
This shows that the correction term can change significantly the motivic properties of
the hypersurfaces, with respect to the classical case.
\end{abstract}

\section{Introduction} \label{intro}
In perturbative quantum field theory, Feynman diagrams are widely used to calculate probability amplitudes. 
Numerical computations on Feynman diagrams show 
that multiple zeta values also occur in perturbative quantum field theory \cite{Bloch07}, \cite{BEK06}. 
For example, the wheel with $n + 1$ spokes gives us $\zeta(2n - 1)$, the Riemann zeta value of odd integers 
\cite{BK97}. 

On the other hand, parametric representation of Feynman integrals provides a powerful tool 
to study the renormalizability and other analytical properties of Feynman graphs.
In \cite{GR07}, \cite{KRTW10}, \cite{Tanasa08}, the authors derived the parametric representation 
of the noncommutative field theory defined by the Moyal product and it turned out to be 
a renormalizable field theory.

Recently, researches are also focusing on the relation between parametric representation of Feynman diagrams 
and motive theory of algebraic varieties \cite{M09}. This is inspired by the remarkable fact that 
multiple zeta values are periods of mixed Tate motives over $\mathbb{Z}$ 
\cite{Brown06}, \cite{GM99}, \cite{Terasoma02}.
In the parametric representation, each Feynman graph $\Gamma$ gives rise to
a graph hypersurface variety $X_\Gamma$, 
and its corresponding motive is called a Feynman motive. One would like to know whether
$[X_\Gamma]$ has value in the mixed Tate subring $\mathbb{Z}(\mathbb{L})$ of the Grothendieck ring of algebraic varieties and delivers a mixed Tate period.

In this paper, we consider a noncommutative field theory studied in 3d loop quantum gravity and derive its 
parametric representation in a $\kappa$-linear approximation. A modified graph polynomial is defined based on the parametric representation, whose no 
gravity limit is just the classical first Kirchhoff polynomial. The motivic property of the graph hypersurface of
the tetrahedron is discussed later using computer calculation.  

This paper is organized as follows. In section \S \ref{intro}, we first recall the noncommutative field theory defined by a new star product, whose perturbative expansion gives the same generalized Feynman graph evaluation as the transition amplitude from the Ponzano-Regge model coupled with matter, more details can be found in \cite{FL0601}, \cite{FL0602}.
 
After recalling the original definition of the Kirchhoff-Synmanzik graph polynomials in section \S \ref{ParaRep}, we derive the parametric representation of the noncommutative field theory and introduce a new graph polynomial with quantum gravity $\kappa$-correction, obtained in
a $\kappa$-linear approximation.

In order to let the reader get familiar with the concept of motives, some basic definitions and examples are given as an introduction in section \S \ref{Motive}. Then we use the hypersurface of the tetrahedron defined by the new graph polynomial as an important example to check the number of $ \mathbb{F}_q $-rational points of this hypersurface for several finite fields, it turns out that the hypersurface of the tetrahedron is not polynomially countable.  Since in the classical case
the graph hypersurface of the tetrahedron is known to be polynomially countable and
a mixed Tate motive, this shows that the presence of quantum gravity corrections can significantly alter the motivic structure of the graph hypersurfaces.

\subsection{A new noncommutative field theory} \label{subsec 1.1}

In \cite{FL0602} Freidel and Livine proposed an effective field theory to describe 3d quantum gravity by introducing a new noncommutative star product. In this subsection, we will recall some properties of this star product and the noncommutative field theory. More physical background about 3d quantum gravity will be given in the next subsection.

First define the deformation parameter $\kappa := 4\pi G$, where $G$ is the Newton's gravitational constant. Then the deficit angle $\theta := \kappa m$ is resulted from the insertion of a particle with mass $m$ into spacetime. 
Viewed as a topological defect, the particle creates a conical singularity with deficit angle $\theta$ in space time since a spinless particle is a source of curvature.

To define the star product, one needs a new group Fourier transform, introduced in \cite{FL0602}, mapping
functions on $SO(3)$ to functions on $\mathbb{R}^3$ bounded by $\kappa^{-1}$, 
\begin{equation}
  \begin{array}{ccll}
    F: & C(SO(3)) & \rightarrow & C_\kappa(\mathbb{R}^3) \\
       & \tilde{\phi}(g)  & \mapsto & \phi(X) = \int_{SO(3)} dg ~e^{\frac{i}{2\kappa}tr(Xg)}\tilde{\phi}(g) 
   \end{array}
\end{equation}
where $X \in \mathfrak{so}(3) \simeq \mathbb{R}^3$ and $tr(\cdot)$ is the usual trace function on matrices. 
Using $SO(3)$ instead of $SU(2)$ is a little easier in real calculation and such
strategy is often applied in 3d quantum gravity.
Then the deformed $\star$-product of fields can be defined by Fourier modes
\begin{equation}
\phi \star \psi (X) := \int dg_1dg_2 ~e^{\frac{i}{2\kappa}tr(Xg_1g_2)}\tilde{\phi}(g_1)\tilde{\psi}(g_2)
\end{equation}
In essential, it can be equivalently defined on phases
\begin{equation} \label{starprod}
e^{\frac{i}{2\kappa}tr(Xg_1)} \star e^{\frac{i}{2\kappa}tr(Xg_2)} = e^{\frac{i}{2\kappa}tr(Xg_1g_2)}
\end{equation}
This star product is an associative noncommutative product  \emph{inequivalent} to the famous Moyal product.
In addition, the inverse group Fourier transform is given by
\begin{equation}
\tilde{\phi}(g) = \int  _{\mathfrak{so}(3)}  \frac{d^3X}{8\pi \kappa ^3} \phi(X) \star e^{\frac{i}{2 \kappa} tr(Xg^{-1})}
\end{equation}

Recall that elements $g \in SO(3)$ can be decomposed as
\begin{equation}
g = P_4(g) + i\kappa P^i(g)\sigma_i, \quad P_4^2(g) + \kappa ^2P^i(g)P_i(g) = 1, \quad P_4(g) \geq 0
\end{equation}
where $\sigma_i,~ i= 1,2,3$ are the Pauli matrices. Or let $\vec{P}(g) = (P_1(g), P_2(g), P_3(g))$, then 
\begin{equation}
g = \sqrt{1 - \kappa ^2 |\vec{P}(g)|^2} + \kappa \vec{P}(g) \cdot i\vec{\sigma}
\end{equation}
By the restriction $P_4 \geq 0$, the $3$-vector $\vec{P}(g)$ belongs to a $3$-Ball bounded above, i.e.  $\vec{P}(g) \in B_\kappa ^3 =\{ \vec{P} : |\vec{P}| \leq \kappa^{-1} \}$.
Based on the isomorphism between $\mathfrak{so}(3) \cong \mathfrak{su}(2) = \text{span} \{ i\sigma_1, i\sigma_2, i \sigma_3 \} $ and the projection $\vec{P}(g)= \frac{1}{2i \kappa}tr(g \vec{\sigma})$, one can write the trace as an inner product 
\begin{equation}
e^{\frac{i}{2 \kappa} tr(Xg)} = e^{\frac{i}{2 \kappa} 2 X \cdot \kappa P(g)} = e^{ i X \cdot P(g)}
\end{equation}  
Thus one can write the star product \eqref{starprod} in terms of $P_i \equiv P(g_i)$:
\begin{equation} \label{newstarprod}
e^{ i X \cdot P_1} \star e^{ i X \cdot P_2} = e^{ i X \cdot (P_1 \oplus P_2)}
\end{equation}
here one introduces a new addition between $3$-vectors in $B_\kappa ^3$ as in \cite{FL0602}
\begin{equation}
P_1 \oplus P_2 := P_1 + P_2 -\kappa P_1 \times P_2
\end{equation}
with $\times$ the usual cross product. One remark is in order, notice that the commutator (Lie bracket) in $\mathfrak{su}(2)$ is the same as the cross product in $\mathbb{R}^3$, i.e. 
\begin{equation}
[\kappa P_1^i\sigma _i, \kappa P_2 ^j \sigma _j ] = 2 i \kappa ^2 P_1^i P_2^j \varepsilon_{ijk}\sigma_k, 
\quad  P_1 \times P_2 = \varepsilon_{ijk} P_1^i P_2^j 
\end{equation} 
in a sense, \eqref{starprod} or \eqref{newstarprod} is a natural generalization of the Baker-Campbell-Hausdorff formula. 

The action functional (interaction Lagrangian) was introduced in \cite{FL0602}:
\begin{equation} \label{FldthryX}
S[\phi] = \int \frac{d^3X}{8 \pi \kappa ^3} \big[ \frac{1}{2}(\partial_i\phi \star \partial_i\phi )(X) - \frac{1}{2} \frac{\sin^2m\kappa}{\kappa ^2} (\phi \star \phi)(X) + \frac{\lambda}{3!} (\phi \star \phi \star \phi)(X) \big]
\end{equation}
or in the momentum representation
\begin{equation} \label{FldthryP}
S[\phi] = \frac{1}{2}  \int dg \big(P^2(g) - \frac{\sin^2m\kappa}{\kappa ^2} \big) \tilde{\phi}(g)\tilde{\phi}(g^{-1}) + \frac{\lambda}{3!} \int \prod_{i=1}^3 dg_i \delta(g_1g_2g_3)\prod_{i=1}^3 \tilde{\phi}(g_i) 
\end{equation}
This is our effective field theory describing the dynamics of the matter field after
integrating out the gravitational sector.
This action functional is invariant under the action of a $\kappa$-deformed Poincar{\'e} group \cite{FL0602}. We then can do the perturbative expansion as before and calculate the quantum corrections to the classical theory using Feynman diagrams. 
 
\subsection{The Ponzano-Regge model } \label{subsec 1.2}

Spin foam models provide a ``path integral'' formulation of quantum gravity based on sum-over-two-complexes with representations and intertwiners. The Ponzano-Regge model is a spin foam model of 3d quantum gravity, whose representations are the unitary irreducibles of $SU(2)$, intertwiners are trivial and vertex amplitudes are the $6j$ symbols.

Fix a triangulation $\Delta$ of a $ 3d $ spacetime manifold and consider an arbitrary oriented graph $\Gamma \subset \Delta$ embedded in a surface $S$, in this paper we only consider spherical graphs $\Gamma$, for each edge $e \in E(\Gamma)$ we insert particles with deficit angle $\theta_e$. 

With matter coupled to 3d quantum gravity, the partition function of 
 the spherical graph $\Gamma$ is a state sum model (\cite{FL0602}),
 
\begin{equation}
\mathcal{I}_\Delta(\Gamma, \theta) = \sum_{ j_e  } \prod_{e \notin \Gamma} d_{j_e} \prod_{e \in \Gamma} \chi_{j_e}(h_{\theta_e}) \prod_{t} \left\{
  \begin{array}{l l l}
    j_{e_1} & j_{e_2} & j_{e_3} \\
    j_{e_4} & j_{e_5} & j_{e_6} \\
  \end{array} \right\}
\end{equation}
where the summation is over spin representations assigned to all edges $e \in \Delta$ and the product of the $6j$ symbols is over all tetrahedra $t$ in $\Delta$. In addition, $d_j = 2j +1$ is the dimension of the  spin-$j$ representation and the character $ \chi_{j}(h_{\theta}) = \sin d_j \theta / \sin \theta $. 

This amplitude $ \mathcal{I}_\Delta(\Gamma, \theta) $ is independent of the triangulation $\Delta$ and only depends on the topology of $\Gamma$, so we simply denote it as $\mathcal{I}(\Gamma, \theta)$ from now on. Taking the dual graphs and changing the variables, one obtains the spin foam amplitude in the Ponzano-Regge model.

It is known that the 6j symbol squared can be written as a group integral.
After getting rid of the $6j$ symbols, one obtains the Feynman graph evaluation:
\begin{equation}
  \mathcal{I}(\Gamma, \theta) = \int \prod_{e \in \Gamma} dG_e \Delta(\theta_e) \delta_{\theta_e}(G_e)
                                         \prod_v \delta(G_v) 
\end{equation}
In order to get this formula, one changes the variable $G_e = g_{t(e)} g_{s(e)}^{-1}$, that is,
 the group element associated to each edge 
$e \in \Gamma$ with $t(e)$, $s(e)$ being the target and source of $e$
and the ordered product of the edge group elements meeting at $v$ is defined as
$G_v = \overrightarrow{\prod}_{ e \supset v} G_e^{\epsilon_v(e)}$ where $\epsilon_v (e) = \pm 1$  
depends on whether $v$ is the target or source vertex of $e$.
 In addition, $\Delta(\theta) = \sin (\theta)$ and the distribution
$\delta_\theta(g)$ is defined by
\begin{equation*}
  \int_{SO(3)} dg \, f(g)\delta_\theta(g) = \int_{SO(3)/U(1)} dx f(x h_\theta x^{-1})
\end{equation*}
where $h_\theta = \bigl(\begin{smallmatrix}
e^{i\theta}&0\\ 0& e^{-i\theta}
\end{smallmatrix} \bigr)$ generates the subgroup $U(1)$ and $dg$, $dx$ are normalized invariant measures.

The delta function on $SO(3)$ can be expanded as
\begin{equation*}
  \delta(g) = \frac{1}{8\pi \kappa^3} \int_{\mathfrak{so}(3)} d^3X \,e^{\frac{i}{2 \kappa} tr(Xg)}
\end{equation*}
Thus the non-abelian Feynman graph evaluation reads,
\begin{equation}
   \mathcal{I}(\Gamma, \theta) =  \int \prod_{v \in \Gamma} 
\frac{d^3X_v}{8\pi \kappa ^3} \int \prod_{e \in \Gamma}dG_e\Delta(\theta_e) \delta_{\theta_e}(G_e) \prod_{v \in \Gamma} e^{\frac{i}{2 \kappa} tr(X_vG_v)}
\end{equation}
One would like to split the exponentials
$exp\{ \frac{i}{2\kappa} tr(X_v \prod_e G_e^{\epsilon_{v}(e)})\}$  and write the evaluation as a
standard Feynman amplitude. This can be done by the noncommutative star product, i.e.
we change the usual product by the star product in the last term.
\begin{equation}
   \mathcal{I}(\Gamma, \theta) =  \int \prod_{v \in \Gamma} 
\frac{d^3X_v}{8\pi \kappa ^3} \int \prod_{e \in \Gamma}dG_e\Delta(\theta_e) \delta_{\theta_e}(G_e) \prod_{v \in \Gamma} 
( \star_{e \in \partial v} e^{\frac{i}{2 \kappa} tr(X_vG_e^{\epsilon_v(e)})} )
\end{equation}

With $2i \kappa \vec{P}(g) = tr(g\vec{\sigma})$, or simply $\vec{P} \equiv \vec{P}(g)$, the distribution  $\delta_\theta(g)$ in the momentum space is
\begin{equation}
\delta_\theta(g) = \frac{\pi}{2 \kappa}\frac{\cos\theta}{ \sin \theta} \, 
 \delta ({|P|^2 - \frac{\sin^2 \theta}{\kappa ^2} }  )
\end{equation}

It turns out that these spin foam amplitudes are exactly the generalized Feynman diagram evaluations of 
the noncommutative field theory described in the previous subsection.

The perturbative expansion of the noncommutative field theory \eqref{FldthryX} or \eqref{FldthryP} gives the sum of spin foam amplitudes $\mathcal{I}(\Gamma, \theta)$ over trivalent diagrams:
\begin{equation}
 \sum_{\Gamma ~trivalent} \frac{\lambda^{V}}{S_\Gamma}\mathcal{I}(\Gamma, \theta)
\end{equation}
where $\lambda$ is a coupling constant, $ V \equiv \sharp V(\Gamma) $ is the number of vertices of $ \Gamma $ and $ S_\Gamma $ is the symmetry factor of the graph.

For a spherical graph $\Gamma$, the spin foam amplitude spells out as
\begin{equation}
\mathcal{I}(\Gamma, \theta) = \prod_{e \in \Gamma}(\frac{\cos \kappa m_e}{4\kappa }) \int \prod_{v \in \Gamma} 
\frac{d^3X_v}{8\pi \kappa ^3} \int \prod_{e \in \Gamma}dG_e \frac{i}{|P_e|^2 - \frac{\sin^2m_e\kappa}{\kappa ^2}} \prod_{v \in \Gamma} \big( \star_{e \in \partial v} e^{i \varepsilon_{ve}X_v \cdot P_e}\big)
\end{equation}
with $P_e \equiv P(G_e)$, one can find the derivation of this formula in \cite{FL0602}.  

Any element $G_e \in SU(2)$ can be written in terms of a Lie algebra element
\begin{equation}
G_e \equiv e^{\kappa P_e \cdot i\vec{\sigma}} = \cos(\kappa |\vec{P}_e|) + i \sin(\kappa |\vec{P}_e|) \vec{n}\cdot \vec{\sigma} 
\end{equation}
where $\vec{n}$ is the direction of the rotation and $\vec{\sigma}$ the Pauli matrices.
Combine this with the normalization 
\begin{equation}
 \int dG_e = \frac{\kappa ^3}{\pi ^2} \int_{B_\kappa^3} \frac{d^3P_e}{\sqrt{1- \kappa ^2 |P_e|^2}},
\end{equation}
up to some constant, one obtains the spin foam amplitude
\begin{equation} 
\mathcal{I}(\Gamma, \theta) =  \int_{\mathfrak{so}(3)} \prod_{v \in \Gamma} 
{d^3X_v}\int_{B_\kappa^3} \prod_{e \in \Gamma} \frac{ d^3P_e}{\sqrt{1- \kappa ^2 |P_e|^2}} \frac{i \cos \kappa m_e }{ |P_e|^2 - \frac{\sin^2 \kappa m_e}{ \kappa ^2} } \prod_{v \in \Gamma} \big( \star_{e \in \partial v} e^{i \varepsilon_{ve}X_v \cdot P_e}\big)
\end{equation}
where the vertices $X_v$ are integrated over $\mathfrak{so}(3) \cong \mathbb{R}^3$ and the momenta $P_e$ 
are integrated over $B_\kappa^3$, the $3$-ball bounded by ${\kappa}^{-1}$. For a trivalent graph $\Gamma$, 
the relation between the number of vertices and that of edges is $2\sharp E(\Gamma) = 3 \sharp V(\Gamma)$, 
extra $\kappa$'s have been canceled out in the above amplitude. 

An explicit computation of the amplitude  $\mathcal{I}(\Gamma, \theta)$ is impossible, so we  consider
a first order approximation in $\kappa$. In the integrand, a term like 
${\cos \kappa m_e}/{\sqrt{1- \kappa ^2 |P_e|^2}}$ does not contribute to a linear approximation since it only 
has even order terms in its Taylor expansion. Therefore, 
 the right amplitude we will use to derive the parametric representation is the following:
   \begin{equation}  \label{QGamp}
\mathcal{I}(\Gamma, \theta) =  \int_{\mathfrak{so}(3)} \prod_{v \in \Gamma} 
{d^3X_v}\int_{B_\kappa^3} \prod_{e \in \Gamma} { d^3P_e} \frac{i}{ |P_e|^2 - \frac{\sin^2 \kappa m_e}{ \kappa ^2} } \prod_{v \in \Gamma} \big( \star_{e \in \partial v} e^{i \varepsilon_{ve}X_v \cdot P_e}\big)
\end{equation}

\section{Parametric representation} \label{ParaRep}

In this section we derive the parametric representation in a $\kappa$-linear approximation
of Feynman integrals given by the spin foam amplitude \eqref{QGamp}. We begin by recalling
the parametric representation in the classical case, and then we compute the effect of the
quantum gravity corrections.

\subsection{Parametric Feynman integrals } \label{subsec 2.1}
In this subsection, following the textbook \cite{IZ06}, we will briefly recall how one gets the
 Kirchhoff-Synmanzik graph polynomials from Feynman diagrams, one can find more details with concrete examples in \cite{BD65} and \cite{IZ06}. 

Consider a quantum scalar field theory based on some interaction Lagrangian, one obtains the Feynman graphs by expanding out the exponential of the interaction term in the Feynman path integral. Choose an arbitrary Feynman graph $\Gamma$, one computes the amplitude $\mathcal{I}_\Gamma$ according to the Feynman rules. The contribution $\mathcal{I}_\Gamma$ can be written in a parametric form in terms of the Schwinger parameters.

We are using the standard notations on oriented graphs. First the incidence matrix $\{ \varepsilon_{v e} \}$ of $\Gamma$ is the $\sharp V(\Gamma) \times \sharp E(\Gamma) $ matrix defined as
\begin{equation}
\varepsilon_{v e} = \left \{
  \begin{array}{l l}
    +1 & \quad \text{if $v = t(e)$}\\
    -1 & \quad \text{if $v = s(e)$}\\
    0  & \quad \text{if $v \notin \partial (e)$}\\
  \end{array} \right.
\end{equation}
For each vertex $v \in V(\Gamma)$, the conservation of momentum gives rise to the Dirac delta function at $v$
\begin{equation}
\delta_v(k,p)=\delta(\sum_{i=1}^n \varepsilon_{v e_i} k_{e_i} + \sum_{j=1}^N \varepsilon_{v e_j} p_{e_j})
\end{equation}
where for each internal edge $e_i \in E_{int}(\Gamma)$, the associated momentum is denoted by $k_{e_i}$, similarly, $p_{e_j}$ is used for each external edge $e_j \in E_{ext}(\Gamma)$. As in \cite{IZ06}, denote the total external momentum $P_v := \sum_{j=1}^N \varepsilon_{v e_j} p_{e_j} $ at each vertex. In the above delta function, $n = \sharp E_{int}(\Gamma)$, $N = \sharp E_{ext}(\Gamma)$ and $V = \sharp V(\Gamma)$ for simplicity.  

Then the transition amplitude $ \mathcal{I}_\Gamma $ is derived from Feynman rules
\begin{equation} \label{IGamma}
\mathcal{I}_\Gamma = C(\Gamma) \int \prod_{i=1}^n \frac{d^Dk_{e_i}}{(2\pi)^D}(\frac{i}{k_{e_i}^2 - m^2 +i\epsilon}) \prod_{v =1}^V \delta_v(k_{e_i},p_{e_j}) 
\end{equation}
where $C(\Gamma)$ is some constant relevant to $\Gamma$, we can just drop such constant in computation since they can be easily restored later. As a convention, $i\epsilon$ is always absorbed into $m^2$ in the Feynman propagator.

By introducing the Schwinger parameters $\{ t_e \}$, the Feynman propagator can be expressed as:
\begin{equation}
\frac{i}{k_e^2 - m^2 } = \int_0^\infty dt_e e^{it_e(k_e^2 - m^2)}
\end{equation}
Meanwhile, write the delta function in the integral form by Fourier transform
\begin{equation}
\delta_v(k_{e_i},p_{e_j}) = \int d^Dx_v e^{ix_v \cdot (\sum_{i=1}^n \varepsilon_{v e_i} k_{e_i} + P_v ) }
\end{equation}
Put these together, then integrate over the internal momentum variables $ k_{e_i} $ and the vertex variables $ x_v $, one can rewrite the amplitude as
\begin{equation} \label{Amp}
\mathcal{I}_\Gamma 
  = \delta(\sum_v P_v) \int \prod_{i=1}^n  dt_{e_i} \frac{e^{-it_{e_i} m^2} }{(-2 i t_{e_i})^{D/2} \det(D_\Gamma(t))^{1/2} } e^{- \sum_{v,v'} [D_\Gamma^{-1}(t)]_{v,v'} P_v P_{v'}  } 
\end{equation}
where the $(V-1) \times (V-1)$ matrix is defined as 
\begin{equation}
[D_\Gamma(t)]_{v,v'} = \sum_{i=1}^n  \frac{i}{2t_{e_i}}\varepsilon_{v e_i}  \varepsilon_{v' e_i}
\end{equation}
Each element $[D_\Gamma(t)]_{v,v'}$ can be viewed as a collection of edges connecting distinct vertices $v$ and $v'$ and for such edge $e \in E_{int}(\Gamma)$ one associates the inverse of Schwinger parameter $1/t_e$. Ignoring possible negative signs, the determinant $\det(D_\Gamma(t))  $ gives all possible paths connecting all the vertices. In other words, one can define the first Kirchhoff-Synmanzik  polynomial as 
\begin{equation}
U_\Gamma(t) := \det(D_\Gamma) \prod_{i=1}^n t_{e_i}  = \sum_\mathcal{T}\prod_{e \notin \mathcal{T}} t_e 
\end{equation} 
where $\mathcal{T}$ ranges over spanning trees of $\Gamma$.

The second Kirchhoff-Synmanzik polynomial $V_\Gamma(t,P)$ involves the external momentum $P_v$ and elements of the inverse matrix $[D_\Gamma^{-1}(t)]_{v,v'}$. We don't give the details here, the interested reader can consult the textbook \cite{IZ06}. The second graph polynomial does not appear in our context, which will be explained later.

In sum, the amplitude can be expressed in the parametric form
\begin{equation}
\mathcal{I}_\Gamma(t, P) =  \delta(\sum_v P_v) \int \prod_{i=1}^n  e^{-it_{e_i} m^2} dt_{e_i} \frac{e^{- V_\Gamma(t,P) / U_\Gamma(t) } }{ U_\Gamma^{1/2}(t) } 
\end{equation}

\subsection{Graph polynomial with $\kappa$-correction} \label{subsec 2.2}
In this subsection, we will show how to get the parametric representation of the spin foam amplitude \eqref{QGamp} for a spherical graph $\Gamma$. Notice that external momentum should be gauged away \cite{FL04}, so we only integrate 
over internal momenta. This gauge fixing process kills graph polynomials involving external momentum and we are mainly interested in the generalized first Kirchhoff polynomial.

As before, we introduce the Schwinger parameters $\{ t_e \}$ for each edge $e \in E(\Gamma)$ and the new Feynman propagators are written as
\begin{equation}
\frac{i}{|P_e|^2 - \frac{\sin^2m_e\kappa}{\kappa ^2} } = \int_0^\infty dt_e e^{it_e(P_e^2 - \frac{\sin^2m_e\kappa}{\kappa ^2})}
\end{equation} 
Note that taking the no gravity limit $\kappa \rightarrow 0$ would recover the classical Feynman propagator and reduce the $\kappa$-deformed star product back to the usual product of fields, so it is reasonable to expect that the new Kirchhoff graph polynomial is the usual Kirchhoff polynomial plus some $\kappa$-related correction in quantum gravity. 

Plug the integral form of Feynman propagators into the spin foam amplitude and use the new addition in the momentum space 
\begin{equation}
\begin{array}{ll}
\mathcal{I}(\Gamma, \theta) 
 & =  \int \prod_{v} d^3X_v \int \prod_{e }d^3P_e\int_0^\infty dt_e e^{it_e(P_e^2 - \frac{\sin^2m_e\kappa}{\kappa ^2})}      e^{i \sum_v X_v \cdot (  \oplus_{e \in \partial v}\varepsilon_{ve} P_e )} \\
 & =  \int \prod_e dt_e e^{-it_e \frac{\sin^2m_e\kappa}{\kappa ^2}} \int \prod_{v} d^3X_v \int d^3P_e e^{it_eP_e^2 + i \sum_v X_v \cdot (  \oplus_{e \in \partial v} \varepsilon_{ve} P_e )} 
\end{array}
\end{equation}
Since $\Gamma$ is embedded in the triangulation $\Delta$, there are exactly three incoming or outgoing edges at each vertex.
\begin{equation}
\begin{array}{ll}
\oplus_{e \in \partial v}\varepsilon_{ve} P_e & = \varepsilon_{ve_i} P_{e_i} \oplus \varepsilon_{ve_j} P_{e_j} \oplus \varepsilon_{ve_k} P_{e_k} 
 = \varepsilon_{ve_i} P_{e_i} + \varepsilon_{ve_j} P_{e_j} + \varepsilon_{ve_k} P_{ve_k} \\ 
 & -\kappa \varepsilon_{ve_i}P_{e_i} \times \varepsilon_{ve_j} P_{e_j}  -\kappa \varepsilon_{ve_i} P_{e_i} \times \varepsilon_{ve_k} P_{e_k}  -\kappa \varepsilon_{ve_j} P_{e_j} \times \varepsilon_{ve_k} P_{e_k}
\end{array}
\end{equation}
The associative addition of $3$-momentum  has \emph{no} cyclic symmetry, so we have to fix the ordering of $P_i, P_j, P_k$ as a convention. As usual, we assume an ordering on the edges $e_i, i = 1, 2, \cdots n \equiv \sharp E(\Gamma)$, then the convention is that $i< j < k$ is always satisfied at each vertex. 

Then the last integration on the momentum space is
\begin{equation}
\begin{array}{rl}
I(t,X) = \int \prod_e d^3P_e e^{it_e \sum_e P_e \cdot P_e + i \sum_e \sum_v \varepsilon_{ve} X_v \cdot P_e
- i \kappa \sum_{e \neq e'} \sum_v \varepsilon_{ve} \varepsilon_{ve'} X_v \cdot P_e \times P_{e'} } \\
= \int \prod_e d^3P_e e^{it_e \sum_e P_e^\alpha P_{e \alpha}  - i \kappa \sum_{e < e'} \sum_v \varepsilon_{ve} \varepsilon_{ve'} \varepsilon_{\alpha \beta \gamma} X_v^\alpha P_e^\beta P_{e'}^\gamma + i \sum_e \sum_v \varepsilon_{ve} X_v^\alpha  P_{e \alpha} }
\end{array}
\end{equation}

In the matrix form
\begin{equation}
I(t,X) = \int_{B_\kappa^3 } d\vec{P} e^{-\frac{1}{2}\vec{P}^t A\vec{P} + \vec{J}\cdot \vec{P}}
\end{equation}
where $\vec{P} = ( P_{e_1}^1,  P_{e_1}^2, P_{e_1}^3, \cdots , P_{e_n}^1,  P_{e_n}^2, P_{e_n}^3 )^t$. If we define
\begin{equation}
T_e := \begin{pmatrix}
 -2it_e & 0 & 0 \\ 
  0 & -2it_e  & 0 \\
  0 & 0 & -2it_e 
 \end{pmatrix}
\end{equation}

\begin{equation}
M_{ee'} := 2i  \kappa  \begin{pmatrix}
    \sum_v \varepsilon_{ve} \varepsilon_{ve'} \varepsilon_{\alpha 11} X_v^\alpha 
     &  \sum_v \varepsilon_{ve} \varepsilon_{ve'} \varepsilon_{\alpha 12} X_v^\alpha 
       &  \sum_v \varepsilon_{ve} \varepsilon_{ve'} \varepsilon_{\alpha 13} X_v^\alpha\\
    \sum_v \varepsilon_{ve} \varepsilon_{ve'} \varepsilon_{\alpha 21} X_v^\alpha 
     &  \sum_v \varepsilon_{ve} \varepsilon_{ve'} \varepsilon_{\alpha 22} X_v^\alpha  
       &  \sum_v \varepsilon_{ve} \varepsilon_{ve'} \varepsilon_{\alpha 23} X_v^\alpha \\
    \sum_v \varepsilon_{ve} \varepsilon_{ve'} \varepsilon_{\alpha 31} X_v^\alpha 
     &  \sum_v \varepsilon_{ve} \varepsilon_{ve'} \varepsilon_{\alpha 32} X_v^\alpha 
       &  \sum_v \varepsilon_{ve} \varepsilon_{ve'} \varepsilon_{\alpha 33} X_v^\alpha
 \end{pmatrix}
\end{equation} 
 $M_{ee'}$ is a skew-symmetric matrix, then 

\begin{equation}
A = \begin{pmatrix}
T_{e_1} & M_{e_1e_2}  &   M_{e_1e_3} & \cdots & M_{e_1e_n} \\
 0          & T_{e_2} &   M_{e_2e_3} & \cdots & M_{e_2e_n} \\
\vdots      & \vdots      & \vdots       & \ddots & \vdots     \\
 0          &   0         &    0         & \cdots & T_{e_n}
 \end{pmatrix}_{3n \times 3n}
\end{equation}
$A$ is an upper triangular matrix resulted from our convention on the momenta around each vertex. 

\begin{equation}
J = ( i\sum_v \varepsilon_{ve_1}\vec{X}_v, \cdots, i\sum_v \varepsilon_{ve_n}\vec{X}_v )
\end{equation}

We take the integration over the usual $\mathbb{R}^3$   instead of $B_\kappa^3$, because the error term is obviously less than $\kappa ^3$. To see this, recall that in dimension one, the Gauss error function is defined
as
\begin{equation*}
    erf(x) = \frac{2}{\sqrt{\pi}}\int_{0}^x e^{-t^2} dt.
\end{equation*} and the complementary error function is defined as
\begin{equation*}
  erfc(x)  = 1-erf(x) = \frac{2}{\sqrt{\pi}} \int_x^{\infty} e^{-t^2}\,dt 
\end{equation*}  
For large $x$, the asymptotic expansion of the complementary error function is given by
\begin{equation*}
    \mathrm{erfc}(x) \sim \frac{e^{-x^2}}{x\sqrt{\pi}}\sum_{n=0}^\infty (-1)^n \frac{(2n-1)!!}{(2x^2)^n},\,
\end{equation*}
where $(2n-1)!!$ is the double factorial. We can generalize this formula to higher dimensions and up to a constant 
related to $\pi$, the error term in three dimension is way less than $\kappa ^3$ which will be dropped
 immediately in a linear approximation.

Hence by the Gaussian integral in three dimension, the integral over the momentum $\vec{P}$ is approximated by
\begin{equation}
I(t,X) = \frac{1}{(\det A)^{1/2}} e^{\frac{1}{2}JA^{-1}J^t} = \frac{1}{\prod_{i=1}^n (-2it_{e_i})^{3/2}}e^{\frac{1}{2}JA^{-1}J^t}
\end{equation}
Since $A$ is an upper triangular matrix, its inverse can be computed explicitly. The inverse matrix is another upper triangular matrix, its diagonal elements are just $T_{e_i}^{-1}$ and the sub-diagonal elements are $-T_{e_i}^{-1} M_{e_ie_{i+1}} T_{e_{i+1}}^{-1} $. We write the rest of the elements outside as a higher order correction matrix $\mathrm{O}(\kappa ^2)B$. We will only use linear $\kappa$-approximation and drop higher order terms in the later calculation.
\begin{equation}
T_e^{-1} = \begin{pmatrix}
 i/2t_e & 0 & 0 \\ 
  0 & i/2t_e & 0 \\
  0 & 0 & i/2t_e
 \end{pmatrix}
\end{equation}

\begin{equation}
A^{-1} = \begin{pmatrix}
T_{e_1}^{-1} & - T_{e_1}^{-1} M_{e_1e_2}T_{e_2}^{-1}  &   0  & \cdots & 0 \\
 0          & T_{e_2}^{-1} &  -T_{e_2}^{-1} M_{e_2e_3} T_{e_3}^{-1} & \cdots & 0 \\
  0         &   0          & T_{e_3}^{-1}                           & \cdots & 0 \\
\vdots      & \vdots      & \vdots       & \ddots & \vdots     \\
 0         &   0         & \cdots         & T_{e_3}^{-1}            & -T_{e_{n-1}}^{-1} M_{e_{n-1}e_n} T_{e_n}^{-1}  \\
 0          &   0         &    0         & \cdots & T_{e_n}^{-1}
 \end{pmatrix} + \mathrm{O}(\kappa ^2)B
\end{equation}

Now the spin foam amplitude looks like
\begin{equation}
\mathcal{I}(\Gamma, \theta)  = \int \prod_e dt_e  \frac{e^{-it_e \frac{\sin^2m_e\kappa}{\kappa ^2}} }{ (-2it_{e})^{3/2}}  \int \prod_{v} d^3X_v  e^{\frac{1}{2}JA^{-1}J^t}
\end{equation}
we will go further to take care of the integration over the vertices and define a new graph polynomial.

Let us look at $ \frac{1}{2}JA^{-1}J^t  $ closer

\begin{equation}
\begin{array}{ll}
\frac{1}{2}JA^{-1}J^t & = - \frac{1}{2} \sum_{vv'} \vec{X}_v (\sum_{i=1}^n \varepsilon_{ve_i}T_{e_i}^{-1}\varepsilon_{v'e_i})\vec{X}_{v'}^t +  \\
& ~~~~~~\frac{1}{2}  \sum_{vv'} \vec{X}_v (\sum_{i=1}^{n-1} \varepsilon_{ve_i}T_{e_i}^{-1}M_{e_ie_{i+1}}T_{e_{i+1}}^{-1}\varepsilon_{v'e_{i+1}})\vec{X}_{v'}^t \\
& = - \frac{1}{2} \sum_{vv'} (\sum_{i=1}^n \varepsilon_{ve_i} \frac{i}{2t_{e_i}}\varepsilon_{v'e_i})\vec{X}_v \cdot \vec{X}_{v'} + \\
& ~~~~~~2\kappa  \sum_{vv'v''} (\sum_{i=1}^{n-1} \varepsilon_{ve_i}\frac{i}{2t_{e_i}} \varepsilon_{v''e_i}
 \varepsilon_{v''e_{i+1}}\frac{i}{2t_{e_{i+1}}} \varepsilon_{v'e_{i+1}} )\det(\vec{X}_v, \vec{X}_{v'}, \vec{X}_{v''}) 
\end{array}
\end{equation}

The first part again tells us that if $e$ is some edge connecting two end points then we just associate the parameter $1/t_e$ correspondingly. The usual first Kirchhoff polynomial then can be recovered if we consider all possible spanning trees in the graph $\Gamma$. 

The second part includes a determinant $\det(\vec{X}_v, \vec{X}_{v'}, \vec{X}_{v''}) = \varepsilon_{\alpha \beta \gamma}X_v^\alpha X_{v'}^\beta X_{v''}^\gamma$, we could ignore the possible negative signs in $\varepsilon_{\alpha \beta \gamma}$ as in the textbook \cite{IZ06}. We view the determinant as a means to couple three vertices together. For two adjacent edges connecting these three vertices, we assign a term $ \frac {\kappa} {t_{e_i} t_{e_{i+1}}} $. It can be thought of as a ``quantum tunnelling'' between two vertices connected indirectly. 

If we introduce a new notation $e_i \prec e_j$, which means that edges $e_i$ and $e_j$ are adjacent in some spanning tree $\mathcal{T}$ of $\Gamma$ and $i<j$, then our definition of the generalized first Kirchhoff polynomial is given as follows.
\begin{defn}
The spin foam graph polynomial including a $\kappa$-linear quantum gravity correction can be defined as
\begin{equation} \label{newpoly}
\mathfrak{U}_\Gamma(t) = \sum_{\mathcal{T} \subset \Gamma} (\prod_{e \in \mathcal{T}} t_e + \kappa \prod_{e_i \prec e_j } t_{e_i} t_{e_j})
\end{equation}
\end{defn}
\noindent where $\mathcal{T}$ ranges over those spanning trees of $\Gamma$. Here we use parameters $t_e$ instead of $1/t_e$ because we have to respect the fact that in any spin foam model graphs are already assumed to be dual graphs. It  consists of two parts, the classical part and the correction part.

\begin{lemma}
The graph polynomial $\mathfrak{U}_\Gamma(t)$ is independent of the ordering the edges.
\end{lemma}
\begin{proof}
We only have to take care of the correction term, and there are two subcases due to the triangulation. 

Fix a spanning tree $\mathcal{T}$, those outermost vertices do not contribute to the $\kappa$-correction of the polynomial $\mathfrak{U}_\mathcal{T}(t)$. Among other vertices, if the valence of the vertex is $2$, it gives one pair of adjacent edges, then $t_{e_i}t_{e_j}$ will be produced as a correction term in $\mathfrak{U}_\mathcal{T}(t)$; and if the valence of the vertex is $3$, then it gives three pairs of adjacent edges, which produce a term $t_{e_i}t_{e_j} \cdot t_{e_j}t_{e_k} \cdot t_{e_i}t_{e_k} = t_{e_i}^2t_{e_j}^2t_{e_k}^2$ in the $\kappa$-correction part of $\mathfrak{U}_\mathcal{T}(t)$.

So whatever the ordering of the edges is, the polynomial can be read off directly from combinatorial information of a graph. 

\end{proof}

Let us determine the degree of $\mathfrak{U}_\mathcal{T}(t)$ by introducing a function on vertices: 
\begin{equation}
  d(v) = \left\{
  \begin{array}{l l}
    0 & \quad \text{if val($ v $) = 1}\\
    2 & \quad \text{if val($ v $) = 2}\\
    6 & \quad \text{if val($ v $) = 3}\\
  \end{array} \right.
\end{equation}
We then have  
\begin{equation}
deg(\mathfrak{U}_\mathcal{T}(t)) = \sum_{v \in \mathcal{T}} d(v) ~~\text{and}~~ deg(\mathfrak{U}_\Gamma(t)) = max_\mathcal{T} \{ deg(\mathfrak{U}_\mathcal{T}(t)) \} .
\end{equation}

\subsection{Hamiltonian action at vertices} \label{subsec 2.3}

Two dynamical processes are of special interest for us in the spin foam model. When a Hamiltonian acts on a specific vertex, it could split into a small triangle then expand through spacetime; or in the opposite direction, some small triangle could collapse into a vertex resulted from some action. In both cases, the graph remains trivalent but the triangulation changes.

Let us look at the change of the polynomial $\mathfrak{U}_\Gamma(t)$ in both cases. First suppose we have one vertex $v$ and three edges $a, b, c$ around $v$ in $\Gamma$, after excitation, $v$ splits into a triangle with edges $\alpha, \beta, \gamma$. There are two possible subcases depending on the valence of the vertex in spanning trees $\mathcal{T} \subset \Gamma $.
  
If the valence of the vertex is $3$ in some spanning tree $\mathcal{T}$, that is, assume edges $a, b, c$ are still in $\mathcal{T}$ and the original polynomial of $\mathcal{T}$ is 
\begin{equation}
\mathfrak{U}_0 = 
abcP(t)+ \kappa a^2b^2c^2 Q(t)  
\end{equation} 
where we just use $a,b, c$ instead of $t_a, t_b, t_c$ for simplicity. The polynomials $P(t), Q(t)$ collect the contributions from other parts of the graph.

To get a spanning tree $\tilde{\mathcal{T}}$ based on $\mathcal{T} \cup \alpha \beta \gamma$, we just break the loop $\alpha \beta \gamma $ and we have three possible cases $\tilde{\mathcal{T}} = \mathcal{T} \cup \alpha \beta $, $\tilde{\mathcal{T}} = \mathcal{T} \cup \alpha \gamma $ and $\tilde{\mathcal{T}} = \mathcal{T} \cup \beta \gamma $. Then the new polynomial is
\begin{equation}
 \begin{array}{l l}
 \mathfrak{U}_1 & = abcP(t)(\alpha \beta + \alpha \gamma + \beta \gamma) +  
 \kappa abc Q(t)(a \alpha^3 \beta^3  + c \alpha^3 \gamma^3 + b \beta^3 \gamma^3) \\
& = abcP(t) \nu + \kappa abc  Q(t)(a(\omega / \gamma)^3   + b(\omega /\alpha)^3 + c(\omega / \beta)^3)
 \end{array}
\end{equation}
where we denote $\nu \equiv \alpha \beta + \alpha \gamma + \beta \gamma$ and $\omega \equiv \alpha \beta \gamma $. 

In the inverse process, if the triangle $ \alpha \beta \gamma  $ collapses into a vertex $v$, it is clear that we immediately get $\mathfrak{U}_0$ from $\mathfrak{U}_1$ by the mappings:
\begin{equation}
 \nu \mapsto 1; ~~~~ a(\omega / \gamma)^3   + b(\omega /\alpha)^3 + c(\omega / \beta)^3 \mapsto abc
\end{equation}

Secondly, if the valence of the vertex is $2$ in some other spanning tree $\mathcal{T}'$, assume that passing from $\Gamma$ to 
$\mathcal{T}'$, edge $c$ is deleted and $a, b$ are still in $\mathcal{T}'$. The polynomial related to $\mathcal{T}'$ is
\begin{equation}
\mathfrak{U}_0' = abP'(t)+ \kappa ab Q'(t) 
\end{equation}

Again suppose the vertex $v$ explodes into a triangle $\alpha \beta \gamma$, and the new $\tilde{\mathcal{T}}'$ derived from $\mathcal{T}'$ is just $\tilde{\mathcal{T}}' = \tilde{\mathcal{T}} \cup \alpha \beta $, $\tilde{\mathcal{T}}' = \tilde{\mathcal{T}} \cup \alpha \gamma $ and $\tilde{\mathcal{T}}' = \tilde{\mathcal{T}} \cup \beta \gamma $. It is easy to see the polynomial related to $\tilde{\mathcal{T}}'$ is
\begin{equation}
 \begin{array} {l l}
\mathfrak{U}_1' &= abP'(t)\nu + \kappa abQ'(t)(a \alpha ^3 \beta ^2 + b \alpha ^3 \gamma ^2 + \beta ^2 \gamma ^2) \\
& = abP'(t)\nu + \kappa abQ'(t)(a \alpha (\omega / \gamma)^2 + b \alpha (\omega / \beta)^2 + (\omega / \alpha)^2)
  \end{array}
\end{equation}

Conversely, if the triangle $\alpha \beta \gamma$ shrinks to a vertex $v$ and $v$ has valence $2$ in some spanning tree, we should apply 
\begin{equation}
\nu \mapsto 1; ~~~~ a \alpha (\omega / \gamma)^2 + b \alpha (\omega / \beta)^2 + (\omega / \alpha)^2 \mapsto 1
\end{equation}
and easily get $\mathfrak{U}_0'$ from $\mathfrak{U}_1'$.

\subsection{Deletion-contraction relation} \label{subsec 2.4}

We then obtain the general relation between $\mathfrak{U}_{\Gamma}, \mathfrak{U}_{\Gamma \setminus e}$ and $\mathfrak{U}_{\Gamma / e}$.
\begin{equation}\label{delconeq}
\mathfrak{U}_{\Gamma}= \mathfrak{U}_{\Gamma \setminus e}+ e \cdot \mathfrak{U}_{\Gamma / e}
\end{equation}
where the action by $e$ means for a classical term, we just multiply it by $e$; for a correction term, we multiply it by
$\prod_{e \cap e_k \neq \emptyset} ee_k$.

We explicitly expand the polynomial with respect to $e \equiv t_n$:
\begin{equation}
 \begin{array}{ll}
\mathfrak{U}_{\Gamma} &= \mathfrak{U}_{\Gamma}|_{t_n = 0}  +  \frac{\partial \mathfrak{U}_{\Gamma} }{\partial t_n}|_{t_n =0} t_n + \frac{\partial ^2 \mathfrak{U}_{\Gamma} }{\partial t_n^2}|_{t_n =0} \frac{t_n^2}{2!} +  \frac{\partial ^3 \mathfrak{U}_{\Gamma} }{\partial t_n^3}|_{t_n =0}  \frac{t_n^3}{3!} +  \frac{\partial ^4 \mathfrak{U}_{\Gamma} }{\partial t_n^4}  \frac{t_n^4}{4!}. \\
 & = P_0 + P_1t_n + P_2t_n^2  + P_3t_n^3  + P_4t_n^4 
  \end{array}
\end{equation}
where $P_i(t_1, \cdots, t_{n-1}), i= 0, \cdots 4$ are the $i$-th derivative evaluated at $t_n =0$. The possible biggest power of $t_n$ is $4$ when $e$ connects two vertices of valence $3$ in some spanning tree. 

Notice that this is not a genuine deletion-contraction relation like the one satisfied by the
original Kirchhoff polynomial in the classical case, since here in \eqref{delconeq}, in the
term $e \cdot \mathfrak{U}_{\Gamma / e}$ the action of $e$ is not just multiplication by
the corresponding variable, but it involves also a multiplication by the more complicated term 
$\prod_{e \cap e_k \neq \emptyset} ee_k$. In particular, this means that one cannot directly apply
to this case the techniques developed in \cite{AM1101}, \cite{AM1102}, to study the graph hypersurfaces
based on the deletion-contraction property. In the next section, we use a more direct method,
based on counting points over finite fields, to investigate the properties of the graph
hypersurface of the tetrahedron graph.

\section{Graph hypersurfaces and motives} \label{Motive}

The graph hypersurfaces for the classical Kirchhoff polynomial have been
extensively studied from the point of view of motives and periods (see for
instance \cite{M09} for an overview). In general, one can investigate the
motivic properties of a variety defined over ${\mathbb Z}$ by either computing 
its class in the Grothendieck ring of varieties $K_0({\mathcal V}_{\mathbb Z})$ or 
by considering the reductions modulo primes and computing the number of
points $\# X({\mathbb F}_q)$ over finite fields. We recall here briefly some 
general facts about the Grothendieck ring and the counting of points, and
then we consider explicitly the case of the graph hypersurface of the 
tetrahedron graph, with the quantum gravity correction.

\subsection{The Grothendieck ring of varieties} \label{subsec 3.1}
In this subsection we will  briefly recall some basic facts about the Grothendieck ring of varieties, which offers an useful tool to test the motivic complexity of an algebraic variety.

\begin{defn}
Let $\mathcal{V}_k$ denote the category of algebraic varieties over a field $k$. The Grothendieck ring $K_0(\mathcal{V}_k)$  is the abelian group generated by isomorphism classes $[X]$ of varieties, with the relation 
\begin{equation}
[X] = [Y] + [X \setminus Y]
\end{equation}
for $Y \subset X$ a closed subvariety. The ring structure is defined by fiber product $[X][Y] = [X \times Y] $.
\end{defn}

For varieties defined over ${\mathbb Z}$, one can similarly consider the
Grothendieck ring $K_0(\mathcal{V}_\mathbb{Z})$ with relations as above. In general, one can define the Grothendieck ring of a symmetric monoidal
category.

Roughly speaking, $K_0(\mathcal{V}_k)$ looks like a decomposition of varieties into pieces by ``cut-and-paste''. Let $1 \equiv [\mathbb{A}^0]$ be the motive of a point
and $\mathbb{L} \equiv [\mathbb{A}^1] \in K_0(\mathcal{V}_k)$ be the Lefschetz motive. For example, $[\mathbb{P}^n] = 1+ \mathbb{L} + \cdots + \mathbb{L}^n$ corresponds to the  decomposition of $[\mathbb{P}^n]$ into cells $\mathbb{A}^0 \cup \mathbb{A}^1 \cup \cdots \cup \mathbb{A}^n $.

\medskip

For an algebraic variety, one may think of its class
 $[X]$ in the Grothendieck ring
as a ``universal Euler characteristic" of the variety $X$, see \cite{Bittner04}. 
To explain more precisely what this means, we recall the following notions.

\begin{defn}
An additive invariant of varieties is a map $\chi: \mathcal{V}_k \rightarrow R$, with values in a commutative ring $R$, satisfying \\
(1) Isomorphism invariance: $\chi(X) = \chi(Y)$ if $X \cong Y$ are isomorphic.\\
(2) Inclusion-exclusion: $\chi(X) = \chi(Y) + \chi(X \setminus Y)$ for  $Y \subset X$ closed.\\
(3) Multiplicative: $\chi(X \times Y) = \chi(X)\chi(Y)$
\end{defn}

\begin{examp} 
(1) Topological Euler characteristic: it is the prototype of additive invariants
\begin{equation}
\chi_{top} : \mathcal{V}_\mathbb{C} \rightarrow \mathbb{Z}; \quad X \mapsto \sum_{i \geq 0} (-1)^i dim H^i_c(X(\mathbb{C}), \mathbb{Q}) 
\end{equation} 
(2) Counting points over $\mathbb{F}_q$: this will be useful in the next subsection
\begin{equation}
N_q : \mathcal{V}_{\mathbb{F}_q} \rightarrow \mathbb{Z}; \quad X \mapsto \sharp X(\mathbb{F}_q) 
\end{equation} 
(3) Hodge polynomial: from Deligne's mixed Hodge theory there exists a unique additive invariant
\begin{equation}
P_{Hdg} : \mathcal{V}_\mathbb{C} \rightarrow \mathbb{Z}[u,v];\quad X \mapsto \sum_{p,q \geq 0} (-1)^{p+q} h^{p,q}(X) u^pv^q
\end{equation} 
where $ h^{p,q}(X) = dim H^p(X(\mathbb{C}), \Omega_X^q) $ is the $(p,q)$-Hodge number of $X$.
\end{examp}

Every additive invariant must factor through the Grothendieck ring and induce a ring homomorphism $\chi: K_0(\mathcal{V}_k) \rightarrow R$, which can therefore
be thought of as a ``universal Euler characteristic''.

\begin{examp} Motivic Euler characteristic

Let $\mathcal{M}_k$ denote the pseudo-abelian tensor
 category of pure Chow  motives over a field $k$. 
Objects of $\mathcal{M}_k$  are triples $(X, p, m)$, where $X$ is a smooth projective variety, $p$ is an idempotent correspondence and $m$ an integer. A morphism from $(X, p, m)$ to $(Y, q, n)$ is given by a correspondence of degree $n - m$.

Let $K_0(\mathcal{M}_k)$ denote the Grothendieck ring of the category $\mathcal{M}_k$. There exists \cite{GS96} an additive invariant $\chi :  \mathcal{V}_k \rightarrow K_0(\mathcal{M}_k)$, hence induces a ring homomorphism 
\begin{equation}
\tilde{\chi} :  K_0(\mathcal{V}_k) \rightarrow K_0(\mathcal{M}_k)
\end{equation}  

In $K_0(\mathcal{M}_k)$, we denote $ 1 = [(Spec(k), Id, 0)]$ as the motive of a point and $\mathbb{L} = [(Spec(k), Id, -1)]$ the Lefschetz motive. The Tate motive $\mathbb{Q}(1)$ is defined as the formal inverse of  $\mathbb{L}$, i.e. $\mathbb{Q}(1) = [(Spec(k), Id, 1)]$. Then the Tate twist $m \in \mathbb{Z}$ in the triple $(X, p, m)$ corresponds to the Tate motive $\mathbb{Q}(m) = \mathbb{Q}(1)^{\otimes m}$.

Taking the Tate motives into account, instead of $\tilde{\chi}$ it is natural to consider the
ring homomorphism
\begin{equation}
\chi_{mot} :  K_0(\mathcal{V}_k)[\mathbb{L}^{-1}] \rightarrow K_0(\mathcal{M}_k)
\end{equation}

\end{examp}

As mentioned before, multiple zeta functions are periods of mixed Tate motives over $\mathbb{Z}$, 
so we are interested in the
subring $\mathbb{Z}[\mathbb{L}, \mathbb{L}^{-1}] \subseteq K_0(\mathcal{M}_\mathbb{K})$ for 
a number field $\mathbb{K}$, or equivalently the subring  $\mathbb{Z}[\mathbb{L}] \subseteq K_0(\mathcal{V}_\mathbb{K})$, corresponding to mixed Tate motives generated by 
the Tate objects $\mathbb{Q}(m)$. 

We will work on singular hypersurfaces derived from generalized Feynman graph evaluations. 
It is known \cite{BB03} that
graph hypersurfaces are general enough to generate the Grothendieck ring of varieties,  but many interesting graph hypersurfaces still turn out to be mixed Tate. Although an individual graph is not mixed Tate, the sum over graphs could be mixed Tate \cite{Bloch10}.

We recall the definition of \emph{polynomially countable} based on counting $\mathbb{F}_q$-rational points of $ X_\Gamma $, where the affine hypersurface is defined as
\begin{equation}
X_\Gamma := \{ t \in \mathbb{A}^n ~|~ \mathfrak{U}_{\Gamma}(t) = 0 \}
\end{equation}
with $n = \sharp E(\Gamma)$ for any graph $\Gamma$. 
\begin{defn}
 Consider the function $F_\Gamma : q \mapsto \sharp X_\Gamma(\mathbb{F}_q)$ defined on the set of prime powers $q = p^n$. We say that the graph hypersurface $X_\Gamma$ is polynomially countable if $ F_\Gamma $ is a polynomial in $\mathbb{Z}[q]$.
\end{defn}
Actually, this definition suits for any scheme $X$ of finite type over $\mathbb{Z}$, but we only focus on graph hypersurfaces in this paper. 

With the above definition, the result of \cite{BB03} tells us that $ F_\Gamma $ is not polynomially countable for almost all graphs, and \cite{Bloch10} shows the sum of all  $ F_\Gamma $ for connected graphs with $ n $ edges  is polynomially countable for any $ n \geq 3 $.

\medskip

The relation between the two concepts, mixed Tate and polynomially countable, is a little subtle: for mixed Tate motives $ [X] $, the numbers $\sharp X(\mathbb{F}_p)$ are polynomial in $\mathbb{Z}[p]$ for almost all primes $p$.  The  ``for almost all $p$'' excludes the cases that, 
for instance, the variety may have bad reduction
at finitely many primes and that can alter the behavior of $\sharp X(\mathbb{F}_p)$. Conversely, assuming the Tate conjecture holds and knowing $\sharp X(\mathbb{F}_p)$ are polynomial in $\mathbb{Z}[p]$ for almost all $p$ would imply the motive is mixed Tate.

\subsection{The tetrahedron} \label{subsec 3.2}
As the simplest example in loop quantum gravity, we think of the tetrahedron as a triangulation of the $2$-sphere. The graph hypersurface of the tetrahedron 
 is tested over different finite fields to count the rational points, it turns out that
this hypersurface is general and complicated, which is \emph{not} polynomially countable. 

\begin{figure}

\centering
\resizebox{0.7\textwidth}{!}{%
  \includegraphics{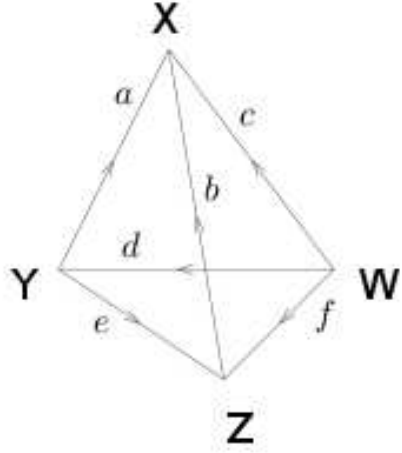}
}
\caption{Tetrahedron}
\label{fig:1}       
\end{figure}
 
\begin{examp}
Let us consider the triangulation $\Delta$ of the 2-sphere as our graph $\Gamma \equiv \Delta$ in the figure. 
Obviously, the tetrahedron has $16$ spanning trees: $abc, abd, abf, ace$, $acf, ade, adf, aef, bcd, bce, 
bde, bdf, bef, cde, cdf, cef$. 

As we can see, the spanning trees can be divided into two types: $T_1$ with one vertex of valence $3$ like $abc$ and $T_2$ with two vertices of valence $2$ like $abd$. Accordingly, there are two types of polynomials $\mathfrak{U}_{T_1}$ and $\mathfrak{U}_{T_2}$, then $\mathfrak{U}_{\Gamma} = \sum_{T_1} \mathfrak{U}_{\mathcal{T}_1} + \sum_{T_2} \mathfrak{U}_{\mathcal{T}_2} $.
It is easy to get the graph polynomial for an individual spanning tree, for instance
\begin{equation}
\mathfrak{U}_{abc} = abc + \kappa a^2b^2c^2; ~~~~ 
\mathfrak{U}_{abd} = abd + \kappa a^2bd
\end{equation}
Now the total graph polynomial is the following:
\begin{equation}
 \begin{array} {ll}
\mathfrak{U}_{\Gamma} 
 & =abc + \kappa a^2b^2c^2 + ade + \kappa a^2d^2e^2 + bef + \kappa b^2 e^2 f^2 + cdf + \kappa c^2 d^2 f^2 \\
 & + abd + \kappa a^2bd + abf + \kappa ab^2f + ace + \kappa a^2ce + acf + \kappa ac^2f \\
 & + adf + \kappa ad^2f + aef + \kappa ae^2f + bcd + \kappa bc^2d + bce + \kappa b^2ce \\
 & + bde + \kappa bde^2 + bdf + \kappa bdf^2 + cde + \kappa cd^2e + cef + \kappa cef^2
 \end{array}
\end{equation}

If we delete some edge, say $e$, then the spanning trees are reduced to $abc, abd, abf$, $acf, 
adf, bcd,  bdf, cdf$. In terms of the graph polynomial, we have $\mathfrak{U}_{\Gamma \setminus e} = \mathfrak{U}_{\Gamma}|_{e = 0}$, that is
\begin{equation}
 \begin{array}{ll}
\mathfrak{U}_{\Gamma \setminus e}
 & =abc + \kappa a^2b^2c^2  + cdf + \kappa c^2 d^2 f^2  + abd + \kappa a^2bd + abf + \kappa ab^2f \\
 & + acf + \kappa ac^2f  + adf + \kappa ad^2f + bcd + \kappa bc^2d + bdf + \kappa bdf^2
  \end{array}
\end{equation}

On the other hand, if we contract the same edge $e$ in $\Gamma$, then $Y$, $Z$ are identified as one vertex and the spanning trees are $ac, ad, af, bc,  bd, bf, cd, cf$ compared to the original $ace, ade, aef, bce, bde, bef, cde, cef$ . In other words, we set $e = 1$ in the classical part and set those adjacent pairs $ee_k = 1$ in the correction part to get the polynomial $\mathfrak{U}_{\Gamma / e}$ . 
\begin{equation}
 \begin{array}{ll}
\mathfrak{U}_{\Gamma /e}
 & = ac + \kappa ac + ad + \kappa ad  + af + \kappa af  + bc + \kappa bc \\
 & + bd + \kappa bd  + bf + \kappa bf + cd + \kappa cd + cf + \kappa cf
  \end{array}
\end{equation} 

We expand the graph polynomial as  $\mathfrak{U}_{\Gamma} = G + eF + e^2 E $, where $G = \mathfrak{U}_{\Gamma \setminus e}$ and 
\begin{equation}
 F  =  ac + ad   + af + bc  + bd  + bf + cd + cf + \kappa a^2c +  \kappa b^2c + \kappa cd^2  + \kappa cf^2
\end{equation}
\begin{equation}
E =  \kappa af + \kappa bd + \kappa a^2d^2 + \kappa b^2f^2 
\end{equation}

\end{examp}

As mentioned in the previous subsection, if an algebraic variety $X$ is mixed Tate as a motive, then its class $[X]$ in the Grothendieck ring of varieties is a  polynomial function with integral
coefficients of the Lefschetz motive $\mathbb{L} =[\mathbb{A}^1]$.

Number of points over finite fields, as an additive invariant, factors through the Grothendieck ring
\begin{equation}
N_q : K_0(\mathcal{V}_{\mathbb{F}_q}) \rightarrow \mathbb{Z}; ~~\text{with}~~ [X] \mapsto \sharp X(\mathbb{F}_q) 
\end{equation} 
Suppose $[X]= \sum_i a_i \mathbb{L}^i \in K_0(\mathcal{V}_{\mathbb{F}_q})$, then (for almost
all primes $p$) the number of
points  $\sharp X(\mathbb{F}_q)$ will be equal to $\sum_i a_i q^i$ since $\mathbb{A}^1$ has $q$ points over $\mathbb{F}_q$. 

When it is hard to compute the Grothendieck class of the graph hypersurface $[X_\Gamma] \in K_0(\mathcal{V}_{\mathbb{F}_q})$, one can instead check whether $X_\Gamma$ is polynomially countable, i.e. verify  $F_\Gamma$ for different prime powers $q$ to see if  $\sharp X_\Gamma(\mathbb{F}_q)$ fits a polynomial $ \sum_i a_i q^i$ for some integral coefficients $a_i$.

The case of the hypersurface of the tetrahedron shows that such $\kappa$-correction terms introduces significant complexity. Without the correction terms, $[X]_{classic} = \mathbb{L}^{5}+\mathbb{L}^{3}-\mathbb{L}^{2}$, showing that the graph hypersurface of
the tetrahedron is a mixed Tate motive. 

However, if we take the correction terms into account,  we obtain a very different result.

\begin{prop}
The hypersurface of the tetrahedron defined by $ \mathfrak{U}_{\Gamma} $ with the parameter $\kappa = 1$ is not polynomially countable.
\end{prop}

\proof By computer calculation, for the fields $\mathbb{F}_p$ with $p = 2, 3, 4, 5, 7, 11$,
we find that the best polynomial fit for the counting of points $\# X_\Gamma(\mathbb{F}_p)$,
by a polynomial with rational coefficient, is given by
\begin{equation}\label{Qpolyn}
\begin{array}{llll}
\# X_\Gamma(\mathbb{F}_p)= & \displaystyle{- \frac {379511}{60480} \, p^5 } & 
  \displaystyle{ + \frac {116827}{576} \,p^4} 
& \displaystyle{ - \frac {24982339}{12096}\, p^3} \\[4mm] & 
\displaystyle{ + \frac {5588389}{576} \, p^2 } 
& \displaystyle{ -\frac {631697377}{30240} \,p } &\displaystyle{ + \frac {1188935}{72} } . \end{array}
\end{equation}

It is proved in Proposition 6.1 of \cite{R06} that if the counting function
$q \mapsto \# X(\mathbb{F}_q)$ of a variety is given by a polynomial function
with rational coefficients, then the polynomial must in fact have integer coefficients.
Thus, the fact that the best polynomial fit for the data $\# X_\Gamma(\mathbb{F}_p)$
for the first few primes $p = 2, 3, 4, 5, 7, 11$ is a polynomial with rational non-integer
coefficients implies that the counting function $q\mapsto \# X_\Gamma(\mathbb{F}_q)$
for the tetrahedron graph must in fact be non-polynomial. 
\endproof

From the relation between mixed Tate and polynomially countable, this proposition suggests that 
the hypersurface of the tetrahedron may also not be mixed Tate. It would be interesting to
see if the Feynman amplitudes of the spin foam model could then be related to explicit
periods that are not mixed Tate periods (multiple zeta values).

\nocite{*}
\bibliographystyle{hplain}
\bibliography{mybib}

\end{document}